\documentclass[aps,pra,reprint,amsmath,amssymb,nofootinbib]{revtex4-2}
\usepackage{graphicx,physics,mathtools,hyperref}
\usepackage{amsmath, amsthm}
\hypersetup{colorlinks=true,allcolors=blue}
\begin{document}

\title{Is Quantum Mechanics a proper subset of Classical Mechanics?}
\author{Khaled Mnaymneh}
\affiliation{Quantum and Nanotechnologies Research Centre, National Research Council Canada, 1200 Montreal Road, Ottawa, Canada K1A 0R6}
\date{\today}

\begin{abstract}
Quantum mechanics is widely regarded as a complete theory, yet we argue it is a tractable projection of a deeper, computationally-inaccessible classical variational structure. By analyzing the coupled partial differential equations of Hamilton's type-1 principal function, we show that classical action-based dynamics are generally undecidable, paralleling spectral gap undecidability in quantum systems. In near-Kolmogorov-Arnol'd-Moser systems, stability hinges on Diophantine conditions that are themselves undecidable, \textit{limiting predictability via arithmetic logic rather than randomness}. Phenomena like spin-3/2 systems and larger, quantum scars and Leggett-inequality violations support this view, naturally explained by time-symmetric classical action. This framework offers a principled resolution to the long-standing dichotomy between unitarity and entanglement by deriving both as emergent features of a tractable rendering from a fundamentally non-separable classical-variational geometry. Collapse and decoherence arise from representational limits, not ontological indeterminism. We propose an explicit experimental test using lateral double quantum dots to detect predicted deviations from standard quantum coherence at the classical chaos threshold. This reframing suggests the classical–quantum boundary is set by computability and not by Planck’s constant. Implications for quantum computing and quantum encryption are discussed.
\end{abstract}

\maketitle

\section{Introduction}
Quantum mechanics is widely regarded as a generalization of classical mechanics, with the latter emerging only in specific limits, such as \(\hbar \to 0\), decoherence, or coarse-grained measurement. Yet the formal derivation of quantum theory is deeply rooted in classical principles, particularly the variational structures underpinning the Hamilton-Jacobi equation (HJE). This paper motivates the view that quantum mechanics does not transcend classical theory, but rather emerges from it as a tractable restriction; one defined by computational, representational, and geometric constraints.

Throughout this work, we adopt an \textit{epistemic stance} unless otherwise stated. By this we mean that quantum mechanics is treated not as a fundamentally new ontological framework, but as an effective formalism that resolves the undecidability or intractability of classical variational equations by collapsing them into decidable, linear, and algebraically closed subspaces. Terms such as ``collapse,'' ``emergence,'' and ``randomness'' refer here not to physical indeterminism, but to epistemic artifacts arising from incomplete access to classical information.

At the heart of classical dynamics lies the principle of stationary action, also known as Hamilton's principle, expressed through Hamilton’s principal function \(S_H(q,t;q_0,t_0)\), which encodes the global structure of trajectories between initial and final boundary data. The formalism gives rise to a pair of coupled partial differential equations that are generically overdetermined and formally undecidable in the general case. Jacobi’s simplification of this structure, reducing it to a single partial differential equation (PDE) solvable by quadrature under integrable conditions, yields the tractable formulation from which the Schrödinger equation naturally emerges via linearization. We contend that quantum theory should be understood as this linear, computable embedding of the richer, but often incomputable, classical phase space.

Our view is motivated by a synthesis of insights from several domains. Schrödinger’s original wave equation was explicitly inspired by the time-independent Hamilton-Jacobi formalism \cite{Schrodinger1926}. Dirac and Feynman extended this connection through operator methods and path integrals \cite{Dirac1933,Feynman:100771}. Koopman and von Neumann recast classical mechanics within a Hilbert space structure \cite{Koopman1931, vonNeumann1932}, while recent work by Field \cite{Field2012} and Simeonov \cite{Simeonov2024} further demonstrates the emergence of the Schrödinger equation from eikonal and action-based perspectives. The Gottesman-Knill theorem \cite{Gottesman1998} reveals that large classes of quantum evolutions, particularly stabilizer circuits, remain classically simulable, pointing to a structural overlap between quantum and classically integrable dynamics.

These observations raise important questions for quantum information theory. If quantum mechanics corresponds to a decidable subspace of a classically undecidable structure, then quantum advantage may be an artifact of our current inability to represent the full complexity of classical chaotic dynamics. Stabilizer circuits, often seen as exceptions within quantum computing, may instead hint at a deeper principle: that quantum speedup only emerges at the computational boundary between classical integrability and undecidability.

In support of this view, we analyze the undecidability of Hamilton’s variational PDEs, link it to the Kolmogorov-Arnol'd-Moser (KAM) theory on number-theoretic conditions, and interpret quantum scars and Leggett inequality violations as empirical signatures of residual classical structure escaping complete quantum encoding. These examples suggest that quantum mechanics, while operationally effective, is representationally incomplete.

The remainder of the paper is structured as follows: Section II introduces the mathematical structure of classical generating functions and the distinction between Hamilton’s type-1 and Jacobi’s type-2 formulations. We demonstrate how solvability and integrability become obstructed in chaotic regimes, with implications for the global representability of classical dynamics. Section III derives the Schrödinger equation as a linearization of an already reduced Jacobi formulation. We show that this transition reflects a loss of global action structure and a restriction to decidable subsets of classical phase space. Section IV explains how time evolution emerges. Section V introduces undecidability and reinterprets core quantum phenomena, including wavefunction collapse, entanglement, and apparent irreducible uncertainty as consequences of representational breakdown rather than intrinsic randomness. Section VI discusses the implications for quantum computing, arguing that quantum advantage is likely a contingent phenomenon arising from our present lack of access to classical models of sufficient complexity. Section VII discusses empirical support. Section VIII outlines a survey of possible continuations beyond standard quantum theory. These include multivalued action functions, nonlinear wave equations, cohomological representations, sheaf-theoretic dynamics, arithmetic and \(p\)-adic quantizations, and deformation quantization, all of which may offer routes toward a more complete variational description, but may have significant consequences for quantum technologies, such as computing and encryption. Section IX offers a discussion on resolving the dichotomy between unitarity and entanglement via emergence and proposes an experimental test of this work. Section X concludes by proposing a research program in \textit{Post-Hamiltonian Representation Theory}, aimed at systematically exploring the logical and geometric limits of both classical and quantum mechanics.

\section{Hamilton's Principal Function}

Let \(q(t)\) be the complete solution, \(f(t;c)\), for some dynamical problem \cite{Calkin1996}, where the $c$'s are the set of $2n$ integration constants determined from initial positions and velocities. This reflects the standard form of classical determinism, where complete knowledge of a system's initial state determines its entire future evolution.

Here, instead, we specify only two positions; the ends of the path in configuration space. The action is then a function of the endpoints, and is known as Hamilton's principal function \cite{Goldstein,Lanczos} (HPF), $S_H(q,t;q_0,t_0)$. It is a type-1 generating function that reframes the problem from an initial-value formulation to a boundary-value formulation. It is not just a technical distinction but reflects a deeper organizational principle of classical dynamics. From this variational perspective, the actual trajectory of the system is not pushed forward by initial conditions, but selected from among possible histories that satisfy endpoint constraints. 

If we now reverse the process, that is, determine an independent way to find \(S_H\), then one only needs to solve two coupled first-order partial differential equations:
\begin{equation}
\frac{\partial S_H}{\partial t_0} = H_0 \quad \text{and} \quad \frac{\partial S_H}{\partial t} = -H,
\label{eq_2}
\end{equation} where $H$ and $H_0$ are the Hamiltonians at the end and start, at times $t$ and $t_0$, respectively. Hamilton introduced his principal function as a tool to reformulate classical mechanics by exploiting the principle of stationary action, encoding the full information of the system’s optimized dynamics between initial and final coordinates. 

The condition of needing to satisfy these two coupled partial differential equations simultaneously, however, imposed a strong duality condition: they demanded that both initial and final configurations be simultaneously specified by the generating function. While elegant in theory, this approach was believed to be analytically intractable for most systems. A proposed simplification \cite{Jacobi1835} essentially abandons one of the partial differential equations and solves for a type-2 generating function\footnote{There are four types of generating functions, however, the arguments found here can be mirrored for type-3 and type-4. Although there are four types, there are two kinds: transformations that cannot (i.e. type-1,4) and can (i.e. type-2,3) be used as identity canonical transformations. There are deep implications when it comes to interpreting what the components of the Taylor series of each kind of generating function means, which will be discussed later.}, \(S_J(q;\alpha;t)\), known as Jacobi's complete integral. The variable \(\alpha\) represents transformed momenta that are constants of motion. This single equation is what is now famously known as the Hamilton-Jacobi equation,
\begin{equation}
    H(q,\frac{\partial S_J(q;\alpha;t)}{\partial q},t)+\frac{\partial S_J(q;\alpha;t)}{\partial t}=0,
\label{eq_2}
\end{equation} where \(H\) is the proposed Hamiltonian. This formulation abandoned the symmetry between \(q\) and \(q_0\) and replaces the overdetermined dual-boundary-value structure of Hamilton’s approach and returned it back to an initial-value problem. The result was not only analytically solvable in many important cases, but it also reframed classical mechanics in terms of evolving a system forward from known initial conditions; an epistemic structure that would later be mirrored in quantum mechanics. 

Jacobi’s insight shifted the emphasis of classical mechanics from global geometric transformations back to a local solvable flow. His type-2 generating function solved the HJE, yielding trajectories via quadrature and was more than technical; it subtly encoded a move from a dual-variable, time-symmetric representation back to a causally directed, one-sided representation of dynamics. This reduction in representational completeness laid the groundwork for the emergence of quantum mechanics.

Jacobi identified several fundamental issues with Hamilton’s insistence on two simultaneous partial differential equations. First, he observed that this framework is mathematically overdetermined and not formally integrable risking internal inconsistency, as partial differential equations must be mutually compatible for a solution to exist. Second, Jacobi pointed out that the second equation is unnecessary; the dynamical problem is completely solved with one and only one of the partial differential equations. Most importantly, Jacobi advanced a tractable and elegant method based on the general theory of integration: by solving a single partial differential equation for $S_J(q;\alpha_i;t)$, with $n$ arbitrary constants $\alpha_i$.

Computationally tractable models necessarily rely on simplifications that prevent them from capturing essential structural aspects of the system. In contrast, intractable or formally nonintegrable formulations, while computationally prohibitive, align more closely with the system’s underlying dynamics. This disparity implies that either Nature harnesses physically inaccessible computational pathways, or that any complete theoretical account must accommodate fundamentally undecidable features \cite{MarkusMeyer1974}. The next two propositions sketch proofs for the integrability of the Hamilton and Jacobi formulation.

\subsection*{Proposition: Integrability in Hamilton’s Type-1 Generating Function Method} 

\noindent Let \(S_H(q, t;q_0,t_0)\) be a smooth type-1 generating function proposed to generate a symplectic transformation \((q_0,p_0)\mapsto(q,p)\), satisfying:

\begin{enumerate}
  \item \(p_i = \partial S_H / \partial q_i\),  \(p_{0i} = -\partial S_H / \partial q_{0i}\)
  \item \(\partial S_H / \partial t + H(q, \partial S_H / \partial q, t) = 0\)
  \item Symplectic preservation: \(dp_i \wedge dq_i = dp_{0i} \wedge dq_{0i}\)
\end{enumerate}

Then for generic, non-separable Hamiltonians \(H(q,p,t)\), the system of PDEs for \(S_H(q, t; q_0,t_0)\) is overdetermined and not formally integrable. \cite{Tabor1989}

\textit{Proof.} The Hamilton-Jacobi equation introduces a nonlinear constraint on \(\partial S_H/\partial q\), while the symplectic preservation condition imposes that the Hessian matrix \(\partial^2 S_H / \partial q_i \partial q_0^j\) must correspond to a closed, exact 2-form generating a symplectomorphism. For non-separable Hamiltonians, these conditions typically fail Cartan’s test for involutivity when prolonged to higher-order jets \cite{BryantGriffiths1995}. As an illustrative case, consider a driven anharmonic oscillator with Hamiltonian \(H = p^2/2 + q^4 + \epsilon q \cos t\). Solving for a global generating function \(S_H(q,t;q_0,t_0)\) with canonical endpoint constraints leads to singularities where characteristics intersect and where no globally smooth function \(S_H\) satisfies both the HJE and symplectic constraints. Hence, type-1 PDE systems are not generally solvable except in the case of separable or integrable dynamics. \qed

\subsection*{Proposition: Integrability in Jacobi’s Type-2 Generating Function Method}

\noindent Let \(S_J(q;\alpha;t)\) be a type-2 generating function solving the reduced Hamilton-Jacobi equation (equation \ref{eq_2}) with \(\alpha_i\) as a set of constants of motion. Then, under regularity conditions on \(H\), \(S_J(q;\alpha;t)\) exists locally and is solvable via the method of characteristics.

\textit{Proof.} Since this is a single first-order PDE in \(q\), with parameters \(\alpha\), the method of characteristics yields a system of ODEs equivalent to Hamilton's equations:
\begin{equation}
\dot{q}_i = \partial H / \partial p_i, \quad \dot{p}_i = -\partial H / \partial q_i, \
\label{eq_4}
\end{equation}
with \(p_i = \partial S_J / \partial q_i\). These can be solved with initial conditions determined by \(\alpha\). 

For example, consider the simple harmonic oscillator with \(H = p^2/2 + q^2/2\). The solution to the HJE is:

\begin{equation}
S_J(q;\alpha;t) = -\alpha t + \sqrt{2\alpha} q \sin t - \frac{1}{2} \sqrt{2\alpha}^2 \int_0^t \cos^2 \tau \, d\tau,
\label{eq_5}
\end{equation} where \(\alpha\) plays the role of the energy and indexes a family of classical trajectories. Thus, \(S_J(q;\alpha; t)\) is recovered via integration and encodes the characteristic solution. The equation is solvable for any smooth \(H\) that allows well-defined characteristics. \qed

These propositions then imply that the second equation in Hamilton’s formalism is not an inconsistency, but a signature of a deeper structure. It encodes a second Hamiltonian flow emerging from the dual time parameters inherent in the full principal function. Unlike Jacobi’s approach, which eliminates this second structure to ensure local solvability and analytic tractability, Hamilton’s formulation retains both the initial and final variables with each generating an independent Hamiltonian evolution. The resulting pair of partial differential equations thus corresponds to two distinct but interrelated Hamiltonians, governing evolution across different foliations of configuration space. While Jacobi’s method simplifies the system into a single effective flow, it does so by sacrificing the global bi-Hamiltonian geometry embedded in the original formulation. Properly understood, the apparent overdetermination in Hamilton’s equations reflects not a flaw but an expression of the richer symplectic structure that emerges when classical dynamics is viewed through the full lens of the principal function.

Jacobi's foundational concern that a function $S_H$ satisfying both partial differential equations might not exist at all can also be addressed, since it is mathematically nontrivial for a single function to satisfy two distinct first-order PDEs simultaneously. Hamilton did not publish such a proof to counter this objection, and his arguments remained largely heuristic. However, later work demonstrated \cite{Conway1940} that such a solution does exist. By carefully selecting boundary conditions and employing canonical transformations within a suitable configuration space, it was shown that a function $S_H$ can be constructed that satisfies both Hamiltonians concurrently. This result affirms that Hamilton’s original formulation, while more geometrically involved, is not mathematically inconsistent. It points to a deeper structural feature of classical mechanics, namely, that under certain conditions, the evolution of the system can be embedded within a higher-dimensional action landscape governed by multiple, compatible Hamiltonian flows. 

It is also an interesting fact, demonstrated by the same work \cite{Conway1940}, that we can achieve the canonical transformation generated by \(S_H\) in two steps by using \(S_J\); this is accomplished by the following transformation:
\begin{equation}
S_H(q,t;q_0,t_0) = S_J(q;\alpha,t)-S_J(q_0;\alpha;t_0).
\label{eq_6}
\end{equation} Figure~\ref{fig:ray} depicts equation \ref{eq_6} applied for different constant momenta in phase space to generate the full path. The various constant momenta can then be understood as a kind of "scaffolding" by the cotangent space of the actual principle function in configuration space.

\begin{figure}[t]
  \centering
  \includegraphics[width=\linewidth]{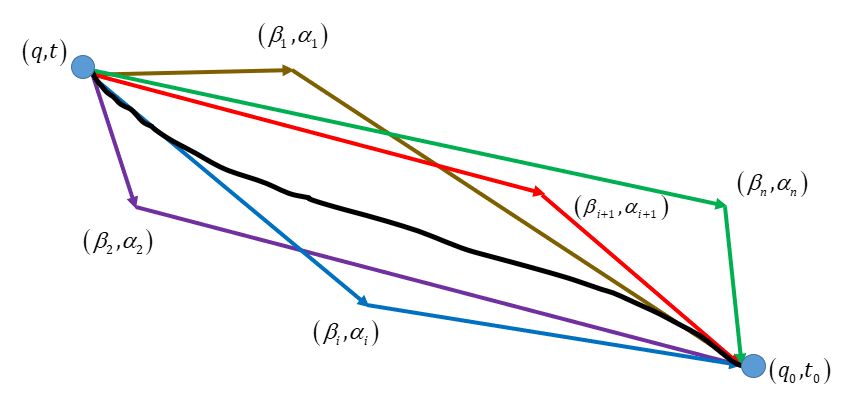}
  \caption{Ray depiction of various \(S_J\) generators with different constant momenta building \(S_H\), the HPF.}
  \label{fig:ray}
\end{figure}

Apart from all that was considered above, it is worth emphasizing that Jacobi’s formulation ultimately prevailed over Hamilton’s original variational approach in practical applications of classical mechanics. While Hamilton introduced the foundational principles, especially the notion of a variational action functional, Jacobi refined these into a more directly solvable framework based on partial differential equations. In particular, Jacobi's scaffolding equation\footnote{The moniker \emph{scaffold} was inspired by equation \ref{eq_6}}, (i.e. the HJE), emerged as the central tool for solving general dynamical problems, providing a single scalar equation whose solution is believed to encode the full dynamics of a system.

Jacobi’s formulation offered several practical advantages. It enabled the reduction of multi-dimensional problems to a single scalar PDE, facilitated separation of variables in many coordinate systems, and provided a natural bridge to canonical transformations and the theory of integrable systems. These features made it a practical and powerful engine of classical analysis.

More importantly, as theoretical physics evolved into the quantum era, the scaffolding equation became a conceptual and mathematical precursor to the Schrödinger equation. Once the appropriate formal analogies were established, particularly via the substitution of the classical action with a complex wave phase and the momentum with operators, the scaffolding equation provided the structural basis upon which the new quantum theory was constructed. In what follows, we will explore this deep correspondence, and show how the quantum wavefunction can be viewed as a simplification of the classical action function.

\section{Structural Descent into Quantum Mechanics}

The introduction of the quantum of action, \(\hbar\), by Planck was originally intended to resolve the mismatch between classical theory and the observed blackbody spectrum \cite{Planck1900}. As quantum theory evolved through the contributions of Schrödinger, Heisenberg, Born, and Dirac, it became clear that existing theoretical tools lacked a unified foundation to connect these distinct frameworks. In this context, the HJE, once central to analytical mechanics and the question of integration, was reexamined and adapted. The recognition that \(\hbar\) naturally aligned with the classical principle of stationary action, particularly through the quantization of action variables, gave new relevance to the Hamilton-Jacobi formalism. Describing \(\hbar\) as the "quantum of action", Sommerfled called it "a most fortunate" name \cite{Sommerfeld1912}, however, this work implies that it may not be as "fortunate" as one thought. In this way, quantum mechanics did not discard classical principles but instead imposed structural constraints, embodied in \(\hbar\), on the classical action function itself. In fact, already before quantum mechanics, stationary action implied that the dynamics of phase space must have a non-zero phase area, where,
\begin{equation}
    \forall \; \boldsymbol{\xi} \neq 0 \; \exists \boldsymbol{\eta} : \omega^2(\boldsymbol\xi,\boldsymbol\eta)\neq0,
\end{equation} is called a symplectic structure on an even-dimensional differentiable manifold \(M^{2n}\) and  \(\omega^2\) is a closed nondegenerate 2-form, \(d\omega^2=0\). The pair \((M^{2n},\omega^2)\) is known as a symplectic manifold (i.e. phase space). 

The integration of the Hamilton-Jacobi equation offers a powerful and, in principle, complete method for solving classical dynamical systems by reducing the problem to that of determining a single scalar function. One way to solve it was inspired by De Broglie's revolutionary hypothesis that particles possess intrinsic wave-like properties. This gave Schrödinger the insight to reformulate the problem of motion in terms of an eigenvalue value question, particularly, 
\begin{equation}
\oint \frac{1}{\lambda}\, dq = n,
\end{equation} where \(\lambda\) is the particle's wavelength, \(q\) is the position coordinate in phase sapce and \(n\) is the integer number of Planck constants needed to revcover the phase space area value.

Rather than seeking a particle’s path through configuration space, he introduced a wave equation whose solutions encode the evolution of complex amplitudes. In doing so, Schrödinger effectively transformed the Hamilton-Jacobi equation into a linear partial differential equation by interpreting the classical action as a phase factor of interfering waves in configuration space modulated by Planck’s constant (see fig.~\ref{fig:interference}). This was enabled by Schrödinger's ansatz of Jacobi's complete integral, \(S_J\).
\begin{equation}
S_J(q;\alpha;t)=K\log\psi(q,t),
\label{eq_7}
\end{equation} where \(\psi\) is the newly introduced wavefunction and \(K\) is the complex proportionality constant. Equation~\ref{eq_7} links the \emph{scaffolding nature} of \(S_J\) and \(\alpha\) to the quantum mechanical wavefunction, \(\psi\), itself.

Thus, quantum mechanics does not emerge by generalizing classical theory but linearizing the scaffolding of an already reduced classical structure. The distinction, between the applications of Hamilton and Jacobi, becomes critical in analyzing the relationship between classical and quantum mechanics. The Schrödinger equation returns to a local time evolution of a wavefunction and discards the global variational structure encoded in \(S_H\). Figure~\ref{fig:jacobi_waveflow} depicts this process in flowchart form, where we end at how the theory is used today. 

\begin{figure}[t]
  \centering
  \includegraphics[width=\linewidth]{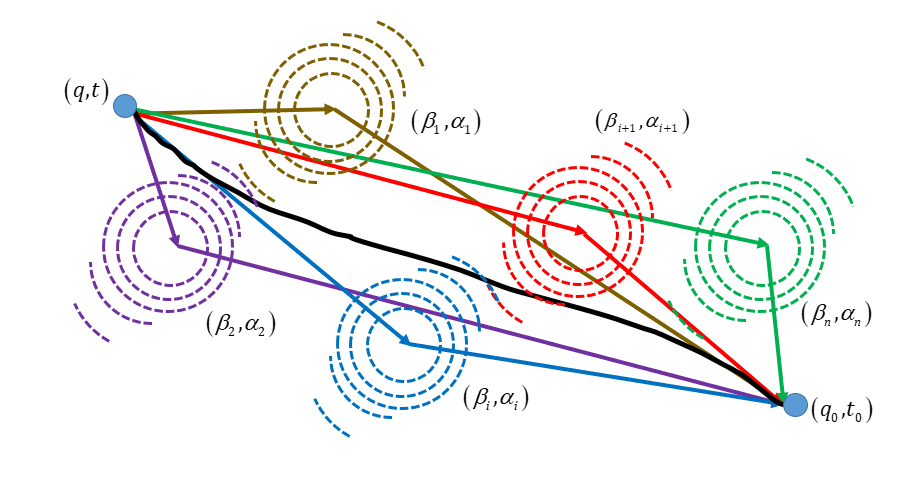}
  \caption{Similar to figure~\ref{fig:ray} but constant momenta now "radiates" out ripples in configuration space than interfere in the shape of the HPF.}
  \label{fig:interference}
\end{figure}

 The foundational structure of quantum mechanics is encoded in Hilbert space, a complete, linear inner-product space that defines how states evolve and interact \cite{vonNeumann1932}. While powerful, this formalism enforces strong global constraints: wavefunctions must be single-valued, square-integrable, and smoothly evolving under unitary dynamics. In contrast, the Hamilton-Jacobi formulation of classical mechanics produces action functions that become multivalued, singular, and non-differentiable in chaotic regimes, where \(S_J\) becomes multivalued or singular due to the formation of caustics. No globally valid solution of the HJE exists \cite{Tabor1989}. Since semiclassical and WKB methods derive quantum dynamics from \(S_J\), quantum mechanics inherits only a regularized approximation of the classical structure, and fails to represent it fully in regimes where the Hamilton-Jacobi framework itself breaks down. The quantum description, while operationally successful in many regimes, cannot contain the full structural richness of classical variational mechanics.

\begin{figure}[t]
  \centering
  \includegraphics[width=\linewidth]{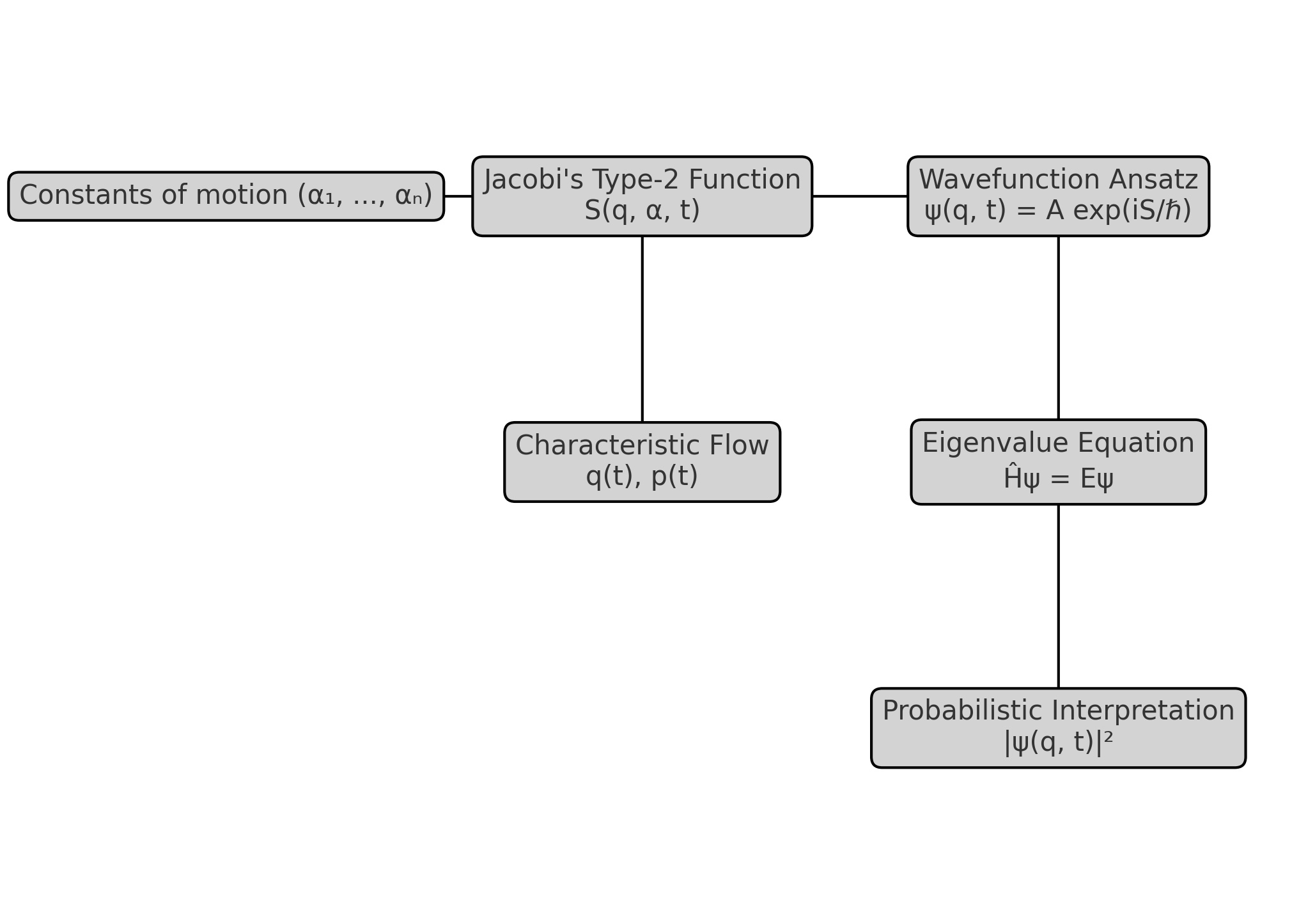}
  \caption{Flowchart illustrating the transition from constants of motion \( \alpha \) in Jacobi's type-2 generating function to the wavefunction \( \psi(q, t) \) and its probabilistic interpretation via eigenvalue decomposition.}
  \label{fig:jacobi_waveflow}
\end{figure}

\section{The Emergence of Time Evolution}

Let \(S_H(q,t;q_0,t_0)\) denote Hamilton’s principal function, which gives the classical action accumulated along a trajectory from initial configuration \(q_0\) at \(t_0\) to the final configuration \(q\) at time \(t\). Suppose this action can be decomposed into two segments, each described by Jacobi's complete integral (i.e., a type-2 generating function):
\begin{equation}
S_H(q,t;q_0,t_0) = S_{J}(q;\alpha;t) - S_{J}(q_0;\alpha;t_0),
\label{eq_9}
\end{equation} where \(\alpha\) represents a shared constant of motion (such as the momentum label), and the composition is made stationary with respect to \(\alpha\) via:
\begin{equation}
\frac{\partial}{\partial \alpha} \left( S_{J}(q;\alpha;t) - S_{J}(q_0;\alpha;t_0) \right) = 0.
\end{equation}

This construction is geometrically analogous to composing two rays in configuration, stitched at an intermediate caustic or surface labeled by \( \alpha \). The wavefronts associated with each ray segment correspond, in the semiclassical regime, to local WKB wavefunctions:
\begin{align}
\psi_{\text{init}}(q_0) &\sim e^{\frac{i}{\hbar} S_{J}(q_0;\alpha;t_0)}, \\
\psi_{\text{final}}(q) &\sim e^{\frac{i}{\hbar} S_{J}(q;\alpha;t)}.
\end{align}
Taking the natural logarithm of both sides of equation \ref{eq_9}, we obtain:
\begin{equation}
e^{\frac{i}{\hbar} S_H(q,t;q_0,t_0)} = \frac{e^{\frac{i}{\hbar} S_{J}(q;\alpha;t)}}{e^{\frac{i}{\hbar} S_{J}(q_0,\alpha,t_0)}} = \frac{\psi_{\text{final}}(q)}{\psi_{\text{init}}(q_0)}.
\end{equation}

Rearranging, we identify this composition as defining the semiclassical time evolution of the wavefunction:
\begin{equation}
\psi(q, t) = e^{\frac{i}{\hbar} S_H(q,t;q_0,t_0)}\psi(q_0, 0).
\end{equation}

Thus, the time evolution operator \(\hat{U}(t)\) in the semiclassical limit is generated by the HPF:
\begin{equation}
\hat{U}(t) = e^{\frac{i}{\hbar} S_H(q,t;q_0,t_0)},
\label{eq_8}
\end{equation} expressed in terms of a type-1 generating function that itself arises from the variationally matched composition of two type-2 complete integrals. 

Taylor expansions of equation \ref{eq_8}, used in gauge theory \cite{Peskin1995} lead directly to non-local influences like the Aharonov-Bohm effect \cite{AharonovBohm1959} and, subsequently, Berry's phase \cite{Berry1984}. It is easy to see then how the global nature of type-1 generating functions can lead to these non-local effects. This perspective further implies that quantum mechanical time evolution is not an imposed structure but rather emerges from the gluing of classical rays across phase space via phase continuity. The wavefronts of constant phase (i.e., level sets of \(S_J\)) evolve via ray congruences that dictate both the quantum phase and the classical trajectory, thereby unifying the geometric and operator-based pictures of evolution.

Although constructing Hamilton's principal function from Jacobi's complete integral offers a structured pathway for solving Hamiltonian systems, it encounters intrinsic limits. The resulting HPF can generate singular canonical transformations in some case \cite{Calkin1996}. This singularity\footnote{Interestingly, the appearance of singularities in the classical construction of the Hamilton's principal function mirrors several phenomena typically associated with quantum mechanics. In particular, the failure of a globally unique mapping in classical phase space corresponds to the emergence of phase singularities, caustic structures, and interference patterns in semiclassical quantum mechanics. Points where the classical action becomes multivalued or singular are precisely where quantum wavefunctions exhibit features such as constructive or destructive interference, phase wrapping, or abrupt changes in probability density. This correspondence suggests that certain "quantum" effects may not be purely quantum at all, but rather artifacts of regularizing classical variational failure modes into a stationary-action, probabilistic formalism. Thus, the structural singularities arising in the time-independent HPF construction provide further evidence that quantum mechanics emerges as a rigid projection of richer classical dynamics.} reflects the failure of a globally unique mapping between configuration space and momentum space, particularly at caustics or regions of multivalued action. Such breakdowns expose the representational incompleteness of the type-2 framework and underscore the need for more comprehensive formulations, such as a more wholistic Hamilton’s type-1 approach, to fully capture classical dynamics in chaotic and global dynamical regimes \cite{ArnoldMechanics}.

\section{Undecidability and other Quantum Puzzles}

The common dismissal of Hamilton’s original two-PDE formulation may warrant reconsideration. What Jacobi regarded as overdetermined may, in fact, reflect a necessary structural redundancy required to capture the global, bi-directional nature of dynamical evolution in complex systems. The risk of internal inconsistency, if the HPF is sought, as Jacobi claimed, is not inconsistency but undecidability. 

\newtheorem{theorem}{Theorem}
\begin{theorem}
There exists no algorithm that, given an arbitrary smooth Hamiltonian $H(q, p, t)$, determines whether the corresponding system of partial differential equations for Hamilton's principal function $S_H(q, t; q_0, t_0)$ has a global classical solution. That is, the general solvability of the Hamilton principal function is undecidable.
\end{theorem}

\begin{proof}[Proof:]
Hamilton's principal function $S_H(q, t; q_0, t_0)$ satisfies the coupled first-order partial differential equations:
\begin{equation}
\frac{\partial S_H}{\partial t_0} = -H(q_0, \frac{\partial S_H}{\partial q_0}, t_0), \quad \frac{\partial S_H}{\partial t} = H(q, \frac{\partial S_H}{\partial q}, t).
\end{equation}
These equations determine the global variational structure of classical mechanics between fixed boundary points $(q_0, t_0)$ and $(q, t)$.

\vspace{1em}
\noindent \textbf{Step 1: Undecidability of PDE Systems.} It is a known result \cite{richardson1968} that the general problem of determining whether an analytic system of partial differential equations has a solution is undecidable. Further, it was demonstrated \cite{pourel1989} that even for linear wave equations, solutions may be noncomputable despite computable initial data. These results establish that the existence of solutions for nonlinear PDEs is, in general, an undecidable problem.

\vspace{1em}
\noindent \textbf{Step 2: Encoding a Turing Machine in a Hamiltonian.} Using methods from computational dynamics and Hamiltonian simulation \cite{moore1990}, one can construct a smooth Hamiltonian $H_M(q, p, t)$ whose dynamics simulate the behavior of a Turing machine $M$ on input $x$. The structure of $H_M$ is chosen such that the existence of a global solution $S_H$ for the associated PDEs is equivalent to the halting of $M$.

\vspace{1em}
\noindent \textbf{Step 3: Reduction to the Halting Problem.} Define the decision problem $\phi_M$: ``Does Turing machine $M$ halt on input $x$?" For the constructed $H_M$, we assert:
\[
\phi_M \text{ is true} \iff S_H \text{ exists globally}.
\]
Since the Halting Problem is undecidable, it follows that there can be no algorithm that decides the existence of global solutions to the system defined by the Hamilton principal function for arbitrary $H$.
\end{proof}

Refer to Appendix A for a formal reduction of this proof-sketch above. This result is related to a recent result \cite{Cubitt2015} about the undecidability of the spectral gap: even with complete knowledge of some Hamiltonian, one cannot, in general, determine whether the system possesses a spectral gap. This undecidability is inherited by quantum mechanics from classical mechanics challenging the notion that quantum theory yields fully predictive, algorithmically representable dynamics. Such results underscores a deeper structural concern: that the wavefunction-based approach of quantum mechanics may be intrinsically incomplete in its capacity to capture global dynamical properties, particularly in systems of high complexity. Incompleteness retains consistency but adds fundamental undecidability in the axiomatic system, specifically here, the axioms of quantum mechanics.

Furthermore, the following theorem expresses the role of the Kolmogorov–Arnold–Moser (KAM) \cite{Kolmogorov1954, Moser1962, Arnold1963} theory and number theory showing classical predictability depends upon undecidable arithmetic properties, thereby supporting to claim that quantum mechanics emerges from a decidable subset of classical dynamics. 

\begin{theorem}
In nearly integrable Hamiltonian systems, the persistence of stable, quasi-periodic trajectories under perturbation depends on number-theoretic conditions that are undecidable in general. Thus, the long-term predictability of classical trajectories is formally obstructed by the undecidability of arithmetic properties of frequency vectors.
\end{theorem}

\begin{proof}
The KAM theorem  states that for a nearly integrable Hamiltonian system, invariant tori survive under small perturbations if their frequency vectors \( \omega \in \mathbb{R}^n \) satisfy a Diophantine condition:
\[
|\omega \cdot k| \geq \frac{\gamma}{|k|^\tau}, \quad \text{for all } k \in \mathbb{Z}^n \setminus \{0\}, \text{ with } \gamma > 0, \ \tau > n-1.
\]
This condition ensures that the frequencies are sufficiently irrational and avoid small denominators that would otherwise destroy the tori.

However, it is known from number theory and computability theory that it is undecidable, in general, whether a given irrational vector \( \omega \) satisfies a Diophantine condition of the above form \cite{davis1982computability}. This undecidability stems from the fact that distinguishing Diophantine irrationals from Louisville numbers, which violate KAM conditions, requires resolving questions about continued fraction expansions that are algorithmically incomputable.

Therefore, although the existence of a quasi-periodic trajectory may be guaranteed in specific cases, the general problem of deciding whether a given orbit lies on a stable torus is formally undecidable. This implies that the long-term stability and predictability of classical systems is not just practically limited, but fundamentally obstructed by number-theoretic undecidability.
\end{proof}

This reinterpretation also offers novel insights into several foundational puzzles in quantum theory, by framing them as artifacts of structural incompleteness that arise when the full classical variational geometry is reduced to a solvable and representationally asymmetric form.

In this view, quantum collapse \cite{vonNeumann1955} need not be seen as a stochastic physical process occuring outside of quantum mechanics, but rather as an emergent consequence of representational breakdown; where the classical action function cannot be globally defined due to logical obstruction (related to the section on spin, later in this section). Quantum theory then appears not as a complete description, but as a computationally constrained subset of a deeper classical structure, with Hamilton’s full formulation offering a potential reintroduction of lost dynamical information.

The notion of \emph{wavefunction collapse} can be understood not as a mysterious physical process but as a byproduct of the epistemic incompleteness of the wavefunction representation. In Jacobi’s type-2 formulation, final configurations are not directly encoded; the wavefunction only contains initial data and constants of motion. Thus, what appears as collapse is simply a discontinuous update to a structure that never had complete boundary data to begin with. This resonates with views from epistemic interpretations of quantum mechanics and with ideas in consistent histories and relational quantum mechanics \cite{Holland1993,Goldstein}.

The phenomenon of \emph{entanglement and nonlocality} arises naturally when we abandon the global representational framework afforded by Hamilton's type-1 structure. In the wavefunction formalism derived from Jacobi's method, there is no direct encoding of full phase space correlations. Instead, correlations between subsystems must be imposed algebraically across tensor products of reduced representations. What appears as nonlocality is a direct consequence of the inability to specify global classical configurations in a system whose description is inherently local and initial-value-based. This complements the notion in algebraic quantum field theory that entanglement is a reflection of global constraints not captured by local observables \cite{Mnaymneh2024} .

Similarly, \emph{quantum randomness} emerges from the fact that the wavefunction only projects a distribution over an ensemble of classical paths labeled by constants of motion. The probabilistic interpretation via \(|\psi|^2\) thus reflects a structural uncertainty, not a fundamental stochasticity. It echoes Born’s ensemble interpretation and more recent decoherence-based derivations of apparent randomness \cite{Holland1993,Koopman1931}.

Even \emph{Planck’s constant} \(\hbar\) can be reinterpreted in this light. Rather than being a fundamental constant of nature, it plays the role of a minimal action scale at which the linearization of classical mechanics become structurally stable. That is, \(\hbar\) sets the lower bound on resolvable areas in phase space consistent with wavefunction-based evolution. This idea finds support in semiclassical mechanics, WKB theory, and geometric quantization, where \(\hbar\) enters as a scaling factor for phase space volume elements and quantization conditions \cite{Goldstein,Lanczos}.

The emergence of intrinsic spin; long treated as a uniquely quantum property, can also be reframed through the lens of structural incompleteness. In classical mechanics, angular momentum is an external, orbital quantity tied to spatial coordinates. However, the regularization of phase-space action into a wavefunction introduces double-valued representations of the rotation group, reflecting the breakdown of a purely classical description of rotational states. 

Spin, in this interpretation, arises not as a new physical entity, but as a structural necessity: a way to encode rotational uncertainty inherent in the loss of classical variational completeness. This view aligns with the geometric interpretation of spinors as sections over nontrivial bundles and suggests that spin, like collapse and entanglement, is a consequence of projecting incomplete classical information into a probabilistic, wave-based formalism.

Spin operators \(S_x, S_y, S_z\) are defined as generators of internal rotation transformations and satisfy the commutation relations:
\begin{equation}
[S_i, S_j] = i \hbar \epsilon_{ijk} S_k,
\end{equation}
where \(\epsilon_{ijk}\) is the Levi-Civita symbol. Regardless of the spin value \(s\), the spin matrices satisfy the Casimir identity:
\begin{equation}
S_x^2 + S_y^2 + S_z^2 = \hbar^2 s(s+1) I,
\label{eq_18}
\end{equation}
where \( I \) is the identity matrix.

For spin-\(1/2\) systems, the spin matrices are simple multiples of the Pauli matrices and their squares trivially commute. Spin-1 operators are not scalar multiples of the identity, however, they still commute with each other allowing the existence of a complete set of common eigenvectors. In these cases, the Casimir identity corresponds directly to relationships among eigenvalues.

In the case of spin-3/2 systems, the component spin-squared operators $S_x^2$, $S_y^2$, and $S_z^2$ each have eigenvalues $9/{4}\hbar^2$ or $1/{4}\hbar^2$. However, these operators do not commute, meaning they cannot be simultaneously diagonalized. As a result, while eq.\ref{eq_18} holds as an operator equation, with all physical states being eigenstates of $S^2$ having eigenvalue $15/{4}\hbar^2$, it does \emph{not} follow that this value is the sum of the eigenvalues of $S_x^2$, $S_y^2$, and $S_z^2$ taken individually. That is, there exists no basis in which all three component squares have well-defined eigenvalues whose sum equals $15/{4}\hbar^2$. Note that this is not a contradiction but a limitation of attempting to assign stochastic-style component decompositions in noncommuting operator contexts \cite{KochenSpecker1967}. 

This highlights a central structural feature of quantum mechanics: even in simple systems, global rotational information encoded in $S^2$ cannot be decomposed into local eigenvalue assignments for individual axes \cite{Ballentine}. In the framework developed here, this limitation is interpreted not as a sign of fundamental indeterminacy, but as a structural consequence of projecting from a multivalued classical action geometry into a decidable, separable formalism. This behavior can be understood as the inevitable consequence of projecting a fixed stationary path into a stationary, linear representation. Spin represents the loss of global rotational determinacy in classical phase space. The necessity of using spinors (representations of SU(2), the double cover of SO(3)) in quantum mechanics reflects the breakdown of classical rotation groups into double-valued, rigid structures.

Taken together, these puzzles, wavefunction collapse, entanglement, randomness, the meaning of \(\hbar\), and the "nonclassicality" of spin, find a common resolution in the representational incompleteness imposed by reducing classical mechanics to a solvable, eigenvalue-structured framework. Rather than viewing quantum mechanics as introducing fundamentally new physics, this perspective shows that its oddities may all stem from constraints and symmetries broken in the process of simplifying classical phase space geometry into a wave representation.

\section{Implications for Quantum Advantage}

In light of this reinterpretation, claims of quantum advantage must be reevaluated. Rather than accessing fundamentally new modes of computation, quantum algorithms may be exploiting regions of classical theory where representational collapse prevents analytic integration. That is, quantum mechanics performs efficiently not because it transcends classical dynamics, but because it approximates and regularizes regions of classical intractability.

The Gottesman-Knill theorem \cite{Gottesman1998} shows that quantum circuits composed entirely of Clifford gates (e.g., Hadamard, CNOT, phase) acting on stabilizer states can be simulated efficiently on a classical computer. These operations correspond to a highly symmetric and algebraically constrained subset of quantum mechanics that maps structurally onto classically tractable transformations, similar to those allowed by Jacobi-type solvability (see figure~\ref{fig:venn}).

True quantum advantage appears to manifest when we move beyond the symmetry-protected regime associated with stabilizer circuits, specifically, with the introduction of non-Clifford elements such as T gates and magic state distillation, which disrupt the linear structure preserved by Clifford operations. Within the standard quantum formalism, these non-Clifford processes break stabilizer symmetries and lead to transformations that cannot be efficiently emulated by known classical means. However, from the perspective developed in this work, where quantum mechanics arises as a linearized, tractable subset of classical variational structure, this boundary may not be absolute. 

The classical simulability demonstrated by the Gottesman-Knill theorem may, in principle, be extended beyond stabilizer circuits. With the development of appropriate non-quantum regularizations, such as frameworks capable of embracing chaotic, multivalued, or sheaf-based structures, it is conceivable that even non-Clifford operations could ultimately admit classical simulation strategies. In this view, quantum computational advantage would not represent a fundamental transcendence over classical mechanics but rather reflect the current absence of complete classical methods for handling nonintegrable, structurally complex transformations. 

Thus, the boundary between classical simulability and quantum speedup is not sharp but structurally defined: it coincides with the breakdown of global classical action formulations and the emergence of probabilistic wave-based representations. Rather than superseding classical mechanics, quantum computation may operate within its unexploited structural degrees of freedom, leaving open the possibility that future classical frameworks, augmented by advances in mathematical modeling or machine learning, could erode or reframe the apparent quantum advantage.

\begin{figure}[t]
  \includegraphics[width=\linewidth]{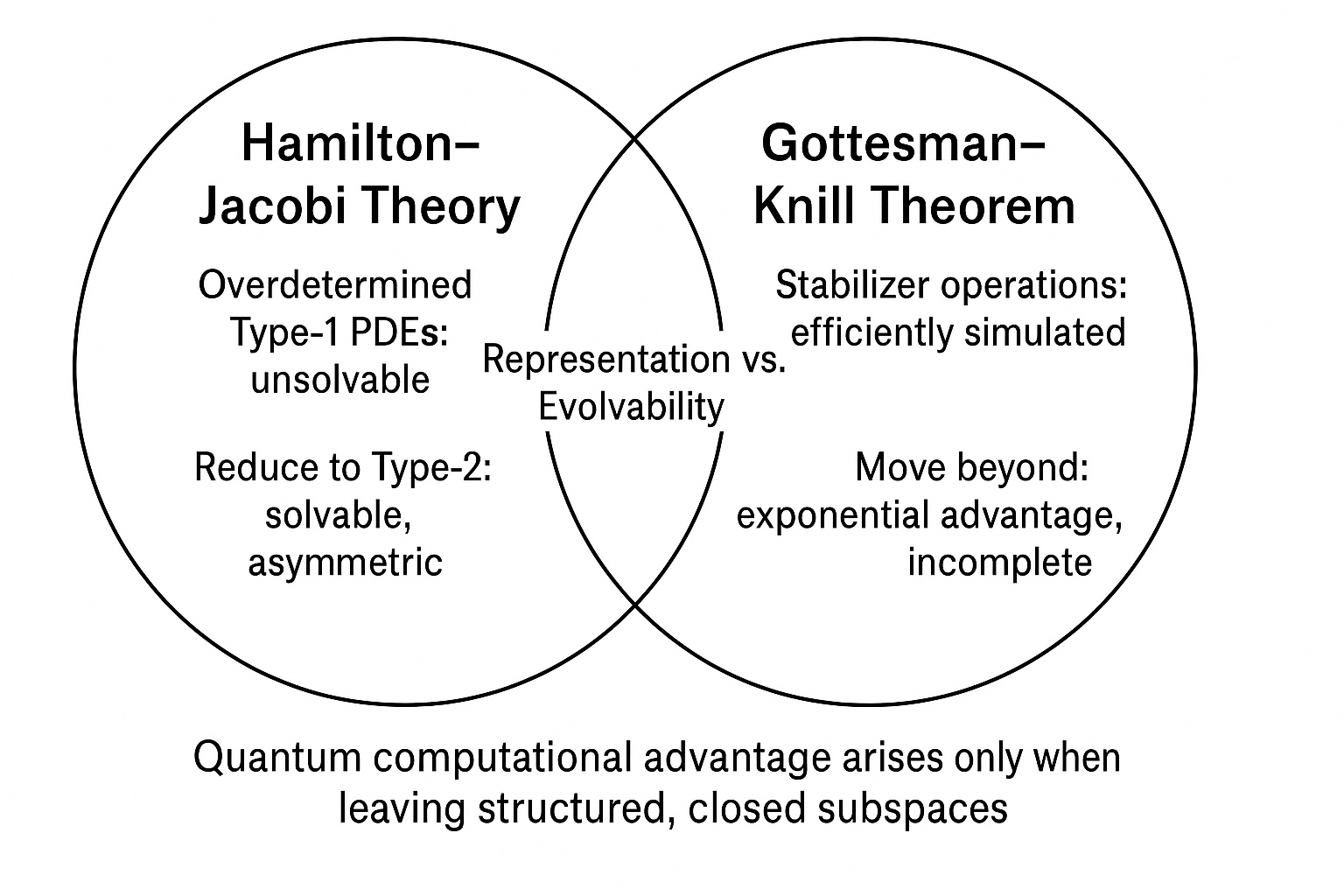}
  \caption{Conceptual overlap between classical type-2 mechanics, Gottesman-Knill quantum operations, and the limitations of representational frameworks.}
  \label{fig:venn}
\end{figure}

This perspective becomes even more compelling when considering chaotic classical systems. In such systems, characterized by sensitivity to initial conditions and exponential divergence of trajectories, the Hamilton-Jacobi equation fails to admit a global smooth solution \cite{ArnoldMechanics}. Characteristics (i.e., classical trajectories) intersect, forming caustics, and the action function \(S_J(q;\alpha;t)\) becomes multivalued or singular, see figrue \ref{fig:caustics_chaos}. This breakdown prevents the construction of a global type-2 solution and, a fortiori, any type-1 formulation.

Quantum mechanics does not resolve classical chaos; it erases it. By replacing the nonlinear dynamics of classical trajectories with the linear evolution of the wavefunction \(\psi(q,t)\), quantum mechanics suppresses the sensitive dependence on initial conditions and the complex structure of caustics that define classical chaotic behavior. There is no true quantum chaos \cite{Berry1989}. The behavior of quantum systems with classically chaotic counterparts reveals a critical asymmetry: while classical mechanics exhibits rich chaotic dynamics, quantum mechanics, bounded by unitarity and spectral discreteness, fails to fully reproduce this complexity. This mismatch highlights the representational limitations of quantum mechanics when describing systems governed by structurally unstable classical dynamics. Our work builds on this observation, interpreting quantum mechanics not as a generalization but as a tractable projection of classical variational principles. The loss of fine phase-space resolution and trajectory-based determinism in quantum formulations can thus be seen as a reflection of this constrained, solvable subset perspective. Quantum mechanics thus stands not as a generalization of classical mechanics but as a proper subset: a stationary, rigid projection that discards the dynamical richness of classical action in favor of linear regularization. We link our arguments about chaos in this section to quantum advantage because chaos is fundamentally nonlinear; and what is understood to move beyond classically-simulable quantum gates need to be highly nonlinear \cite{GKP2001}.

In this framework, the wavefunction is not a faithful compression of classical dynamics but a \emph{rigid projection} that eliminates classical failure modes by construction. The emergence of probabilistic behavior and irreducible uncertainty is not evidence of deeper physical indeterminacy but a symptom of forcing chaotic, multivalued classical flows onto a representationally constrained structure. Quantum mechanics does not regularize chaos; it suppresses it by removing the phase-space complexity that would otherwise manifest. This view aligns with semiclassical studies, where chaotic classical systems exhibit rapid phase scrambling and decoherence-like features without any true quantum analogue of classical instability \cite{Berry1977,Littlejohn1986}.

Thus, quantum advantage and the very necessity of quantum mechanics arise not from an expansion of classical theory but from its \emph{representational breakdown}. When the scaffolding equation (i.e. the HJE) becomes unsolvable due to classical intractability, the linear wave description does not capture the chaotic richness, it replaces it with a stationary, probabilistic formalism. Quantum mechanics succeeds precisely where classical mechanics becomes intractable but at the cost of losing access to the full dynamical complexity of classical action.

\section{Empirical Support for Representational Incompleteness}

Evidence supporting the view that quantum mechanics is a computationally constrained projection of a more complete classical structure can also be found in physical systems where classical dynamics unexpectedly persists within a quantum framework. Three such cases are particularly illuminating: recent experiential demonstration of the violation of the Leggett inequality, quantum scars and the dynamical behavior of Hyperion, Saturn’s chaotic moon.

The Leggett inequality was introduced as a refinement of Bell-type no-go theorems, aimed at excluding a broader class of hidden variable theories that are \emph{non-local} yet still preserve a form of \emph{realism} \cite{Leggett2003}. In contrast to Bell's inequality, which tests local realism, Leggett's framework allows for nonlocal correlations but maintains that measurement outcomes reflect pre-existing, detector-independent properties of the system. Specifically, Leggett-type models retain a statistical independence between hidden variables and measurement settings.

Experimental violations of the Leggett inequality \cite{Groblacher2007} have been interpreted as further restricting the space of viable realist theories. Crucially, these violations do not merely imply nonlocality in the Bell sense, they imply a more radical breakdown: the apparent \emph{influence of the detector setting on the state preparation itself}, suggesting that measurement outcomes cannot be cleanly separated from final boundary conditions.

This observation has significant consequences for interpreting the origin of quantum correlations. In the context of our framework, based on Hamilton’s principal function \(S_H(q, t; q_0, t_0)\), the dynamics of a classical system are determined by a variational principle that is inherently \emph{time-symmetric}. The action depends jointly on both initial and final boundary data. Such a formalism naturally accommodates the kind of apparent retrocausality or global constraint propagation required to account for Leggett-type violations, without appealing to stochastic collapse or abandoning determinism.

Thus, objections to the subsetness of quantum mechanics grounded in \emph{Bell’s theorem}, which assert that no local realist model can reproduce quantum statistics, are insufficient. While Bell rules out locality under a realism assumption, Leggett inequality violations go further: they rule out even certain nonlocal realist models that lack retrocausal or endpoint-symmetric structures. What these results suggest is not that realism must be abandoned, but that the classical structure must include globally constrained dynamics: precisely what is encoded in the type-1 Hamilton principal function.

Accordingly, we interpret the violation of Leggett inequalities as empirical support for the idea that the \emph{measurement apparatus (detector) influences the source}, not dynamically via backward causation, but structurally via the time-symmetric variational constraints encoded in classical action. The failure of local independence assumptions in these experiments reflects the failure of assuming forward-time, measurement-independent dynamics, not the failure of realism itself. Quantum mechanics, under this view, emerges as the tractable, probabilistically packaged approximation to a deeper classical theory governed by globally constrained and potentially undecidable action structure.

Quantum scars are anomalous eigenstates that exhibit unexpected localization in quantum systems whose classical counterparts are chaotic. According to the eigenstate thermalization hypothesis (ETH) and random matrix theory (RMT), the eigenstates of quantum systems with chaotic classical analogs should be ergodic and structureless in the semiclassical limit \cite{Deutsch1991,Srednicki1994,Haake2010}. However, scarred states violate this expectation by remaining concentrated along unstable classical periodic orbits, thereby retaining remnants of classical structure within the quantum formalism.

Originally discovered in single-particle systems such as the stadium billiard \cite{Heller1984}, quantum scars have since been identified in many-body systems, most notably in chains of Rydberg atoms \cite{Bernien2017,Turner2018}. These systems display nonergodic dynamics, including anomalously long-lived revivals and slow thermalization, all associated with a subset of eigenstates that are atypical in their localization and overlap with unstable classical trajectories.

In the context of this paper, quantum scars provide empirical support for the view that quantum mechanics does not fully erase classical variational structure, even in regimes dominated by chaos. The persistence of action-phase coherence along unstable orbits suggests that the Hilbert space formalism is not a complete representation of the classical phase space geometry. In particular, it points to the survival of information related to Hamilton's principal function \( S_H(q, t; q_0, t_0) \), which governs global classical evolution under fixed boundary conditions.

This observation is consistent with our claim that quantum mechanics corresponds to a tractable, decidable subset of classical dynamics. In classically chaotic systems, where KAM tori are destroyed and the global action function becomes intractable or even undecidable, the existence of scarred eigenstates reveals that fragments of the classical structure remain embedded in the quantum spectrum. These are not artifacts of incomplete quantization, but rather markers of an underlying classical reality that resists full projection into Hilbert space.

Thus, quantum scars serve as physical evidence that the quantum formalism may fail to encode the full topology and arithmetic complexity of classical phase space when instabilities are present. They support the broader view advanced in this work: that quantum collapse and randomness emerge not from fundamental indeterminism, but from the computational and representational limitations of the quantum framework in describing a deeper classical variational geometry.

The rotational dynamics of Saturn’s moon Hyperion offer yet another compelling test case at the intersection of classical chaos and quantum theory. Hyperion is an irregularly shaped, gravitating body that exhibits chaotic tumbling motion due to tidal torques from Saturn. Its classical motion is highly sensitive to initial conditions, a hallmark of deterministic chaos in Hamiltonian systems \cite{Wisdom1984}. Small variations in orientation or angular momentum lead to exponentially diverging trajectories, with a Lyapunov time on the order of weeks.

Quantum mechanically, Hyperion’s orientation is described by a rotational wavefunction in angular phase space. According to standard quantum theory, such a wavefunction should spread rapidly due to the underlying classical chaos. As Zurek and Paz have shown \cite{Zurek1995}, the rate of quantum wavefunction spreading in chaotic systems scales with the classical Lyapunov exponent, leading to an effective quantum decoherence timescale of days for macroscopic chaotic bodies like Hyperion.

Yet, despite this prediction, astronomical observations confirm that Hyperion maintains a well-defined classical rotational state over extended periods, with no observable signature of wavefunction delocalization or decoherence. This tension is typically resolved within the quantum literature by appealing to environmental decoherence: Hyperion is constantly interacting with photons, dust, and gravitational fields, which serve to rapidly decohere its quantum state and enforce classicality \cite{Tegmark1993}. However, this explanation presumes the universality and completeness of the quantum formalism, while providing no representational account of how classical trajectories are maintained in the face of exponential wavefunction dispersion. In essence, decoherence explains the suppression of interference terms, but does not resolve the deeper problem of wavefunction collapse or the emergence of definite classical facts.

From the perspective advanced in this work, Hyperion's behavior may instead suggest the continued relevance of classical variational structure in chaotic regimes. If the full classical evolution is governed by a global action function \( S_H(q, t; q_0, t_0) \), then the inability of the quantum formalism to represent this structure in the presence of classical chaos indicates a failure not of physical determinism, but of representational completeness. The case of Hyperion thus supports the broader thesis that quantum mechanics is a tractable but incomplete projection of a deeper classical geometry; one in which collapse and classicality arise from the computational or epistemic inaccessibility of the global action function, rather than from stochastic physical processes.

Together, these examples support the idea that quantum mechanics does not always succeed in erasing classical structure, particularly in systems exhibiting chaos, instability, or non-ergodicity. They provide empirical grounding for the proposal that quantum theory, while effective, is a reduced model derived from a more complete classical framework in which Hamilton’s action principles play a fundamental representational role.

\section{Survey of Alternative Continuations from the Hamilton–Jacobi Equation}

The framework developed in this work suggests that quantum mechanics is not a generalization of classical mechanics but a constrained, rigid projection of classical variational structure. It emerges from the Hamilton–Jacobi equation through a historical and mathematically contingent route: the linearization of classical dynamics via the Schrödinger ansatz. This path sacrifices the full richness of classical action, especially in chaotic regimes, in favor of a stationary, probabilistic, and linear formalism.

In this section, we consider whether quantum theory is merely one of multiple possible projections from the HJE, and whether alternative continuations could preserve more of the underlying classical complexity rather than suppress it \cite{NLab}. Several underexplored avenues in mathematical physics, and even in pure math \cite{Jammer1966}, suggest new frameworks that might retain the more intricate aspects of classical dynamics.

\begin{figure}[t]
  \centering
  \includegraphics[width=\linewidth]{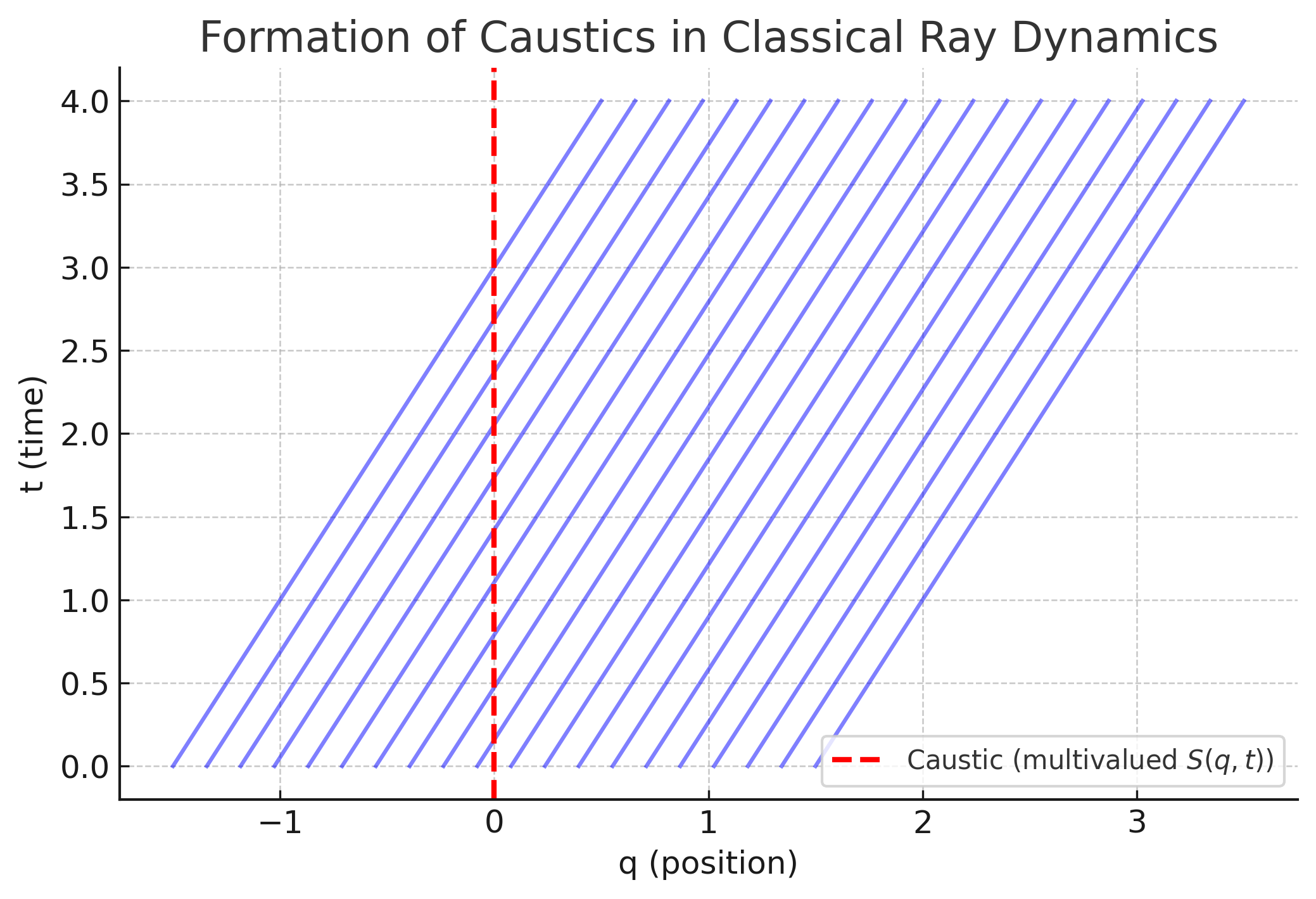}
  \caption{Formation of caustics from intersecting classical trajectories. In the chaotic regime, characteristics cross, causing \( S(q,t) \) to become multivalued. Quantum wavefunctions regularize this behavior by encoding phase rather than action directly.}
  \label{fig:caustics_chaos}
\end{figure}

In the standard construction of quantum mechanics, the Hamilton–Jacobi equation is linearized into the Schrödinger equation:
\begin{equation}
    i\hbar \frac{\partial \psi}{\partial t} = \hat{H}\psi,
\end{equation}
where \(\psi(q,t)\) evolves linearly under the Hamiltonian operator \(\hat{H}\). This linearity is essential: it guarantees the superposition principle, unitary evolution, and the probabilistic interpretation via squared amplitudes. This linearization, however, reflects a representational collapse: an enforced regularization of the multivalued, and often chaotic structure of the original classical variational dynamics. It is natural to ask whether alternative continuations of the Hamilton–Jacobi framework might retain nonlinearity, preserving more of the underlying classical structure rather than projecting it onto a stationary, linear space.

One could consider generalizations where the wavefunction evolves according to a \emph{nonlinear} equation, such as:
\begin{equation}
    i\hbar \frac{\partial \psi}{\partial t} = -\frac{\hbar^2}{2m} \nabla^2 \psi + V(q)\psi + g|\psi|^{2k}\psi,
\end{equation}
where \(g\) and \(k\) are constants controlling the nonlinearity. Such nonlinear Schrödinger-type equations already appear in specific physical contexts: Gross–Pitaevskii equation in Bose–Einstein condensate theory, modeling interacting particle systems with weak nonlinearity (\(k=1\)); nonlinear optics, where the intensity-dependent refractive index induces nonlinear Schrödinger equations for the electromagnetic field envelope; soliton theory, where nonlinearities precisely balance dispersion to produce stable, localized wave packets. Nonlinearity restores certain features lost in pure linear quantum evolution: wave self-interaction, stability of coherent structures (solitons), and modified dispersion and phase evolution.

Pursuing nonlinear wave dynamics as a continuation of classical variational collapse suggests that possible "post-quantum" theories could be built; ones that retain fragments of classical multivaluedness, chaotic sensitivity, and complex phase topology that standard quantum mechanics necessarily suppresses. Nonlinear wave dynamics represents a fundamentally different pathway: not enforcing stationary action through rigid linear projection, but allowing for a controlled, structured reintroduction of complexity into the quantum-classical correspondence.

A central structural constraint of quantum mechanics is its linearity, from which unitarity and the no-cloning theorem both follow. The evolution of quantum states via the Schrödinger equation is linear, and this linearity guarantees that time evolution operators are unitary, preserving inner products and the probabilistic interpretation of the theory. Critically, the no-cloning theorem, which asserts the impossibility of duplicating arbitrary quantum states, relies fundamentally on this linear and unitary structure \cite{Wootters1982}. If quantum mechanics is considered as a linearization of classical theory then the no-cloning prohibition is not a fundamental feature of reality but an artifact of this approximation. In a nonlinear extension of the theory, such as one based on the fully nonlinear variational structure of classical mechanics: unitarity need not hold, and the logic underlying the no-cloning theorem may no longer apply.

In classical mechanics, the action \(S_J(q;\alpha;t)\) plays a central role in characterizing the evolution of systems via the Hamilton–Jacobi equation. In regular regimes, \(S_J(q;\alpha;t)\) is a smooth, single-valued function: at any given point in configuration space \(q\) and time \(t\), there is a unique value of the action, corresponding to a unique classical trajectory. In complex systems, particularly those exhibiting chaotic behavior or intricate phase-space structures, the action becomes inherently multivalued; multiple classical trajectories can arrive at the same point \(q\) at the same time \(t\), each carrying a different classical action value. Thus, \(S_J(q;\alpha;t)\) is no longer a globally single-valued function but becomes a multivalued structure layered over configuration space.

This multivaluedness gives rise to caustics. A caustic is a singular locus in configuration space where multiple branches of the action coalesce, leading to an infinite or undefined gradient of the action. Mathematically, caustics correspond to points where the Jacobian of the transformation from phase-space variables \((q,p)\) to \((q,t)\) vanishes:
\begin{equation}
\det\left(\frac{\partial p}{\partial q}\right) = 0.
\end{equation}

Physically, caustics manifest as regions of enhanced density or interference patterns in classical wave analogues, such as light focusing through a curved surface (e.g., rainbows or optical caustics on the bottom of a swimming pool). While multivalued action refers to the general phenomenon where multiple classical paths intersect a single point in space-time, caustic structures are specific singularities where these multivalued branches fold over one another, creating physical and mathematical singularities. In the traditional construction of quantum mechanics, the presence of caustics and multivalued action is regularized through the wavefunction formalism. Instead of tracking individual classical paths, the wavefunction allows for superposition of amplitudes associated with different classical histories:
\begin{equation}
    \psi(q,t) \sim \sum_j A_j(q,t) \exp\left(\frac{i}{\hbar} S_j(q,t)\right),
\end{equation}
where each \(S_j\) represents a different classical branch of the action.

In this sense, quantum mechanics does not resolve the underlying multivaluedness; it suppresses it into a rigid linear structure, absorbing the complexity into phase interference patterns and probabilistic interpretation. The individual identities of classical trajectories, and their intricate multivalued action structure, are lost in the statistical smoothing of quantum mechanics. If one instead sought to preserve multivalued action and caustic structures explicitly, alternative mathematical frameworks would be required: for instance, describing phase space as a branched manifold, using Riemann surfaces, monodromy representations, or sheaf-theoretic techniques to handle multiple layers of action data systematically. Such an approach would represent a continuation of classical mechanics beyond the collapse into stationary linear quantum theory.

An alternative approach to preserving classical richness is to abandon the assumption of global smoothness altogether. Instead of representing the system with a globally defined function, one can model the evolution of physical quantities using the machinery of sheaf theory and cohomological dynamics. Here, physical observables, trajectories, and even states are treated as local sections over patches of phase space or configuration space. The "global" structure of the system is no longer assumed to exist smoothly everywhere, but rather is assembled from consistent local data according to prescribed gluing conditions. In places where singularities, caustics, or multivaluedness occur, the failure to glue sections smoothly is captured by nontrivial cohomology classes. Instead of a single-valued wavefunction, one might represent the system by local action functionals \( S_i(q,t) \) defined over overlapping patches \(U_i\), transition functions or monodromy data encoding how these patches are related, or cohomological invariants measuring global obstructions to smooth evolution.

Mathematically, the dynamics could be described through derived categories, sheaf cohomology, or even stack-theoretic constructions, where the evolution of the system is governed not by a differential equation on fields, but by morphisms between structured sheaf objects over phase space. Concrete examples of such ideas appear in microlocal analysis, where the singular support of distributions captures phase-space propagation of singularities, topological field theories, where fields are not necessarily smooth but patched together from local data, or sheaf-theoretic models of quantum field theory \cite{freed2010}.

A cohomological or sheaf-based formulation offers a fundamentally richer alternative to standard quantum mechanics; it would allow for the explicit representation of multivaluedness, branching structures, and caustic singularities without collapsing them into smooth probabilistic amplitudes. It could naturally accommodate chaotic behavior, topological phase transitions, and other nontrivial global phenomena, and it might explain phenomena like superselection sectors, phase singularities, or even quantum decoherence as manifestations of cohomological obstructions, rather than fundamental probabilistic behavior. Sheaf-based dynamics represents a path toward preserving the full complexity of classical action structures, offering a radically different continuation of classical mechanics compared to the linear stationary-action collapse that underpins standard quantum theory.

Standard quantum mechanics assumes that configuration space and phase space are built upon the real numbers \(\mathbb{R}\), relying on continuous, differentiable structures. However, classical action-based dynamics, particularly in chaotic regimes, often involve highly intricate, discontinuous, or fractal-like behavior that challenges smooth real-number representations. An alternative path emerges by considering \emph{non-Archimedean} structures \cite{Khrennikov1997}, such as the field of \(p\)-adic numbers \(\mathbb{Q}_p\). The \(p\)-adic numbers offer a radically different notion of "closeness" and "continuity," where proximity is determined not by size but by divisibility properties with respect to a prime \(p\). In a \(p\)-adic framework: numbers can be infinitely close if their difference is highly divisible by \(p\), space becomes ultrametric, obeying a stronger version of the triangle inequality, and structures naturally accommodate branching, hierarchies, and discrete scaling symmetries.

Building on this, \emph{p-adic quantization} proposes to formulate quantum theory over \(\mathbb{Q}_p\) instead of \(\mathbb{R}\). In this setting, the path integral or Hamiltonian evolution could be rewritten as:
\begin{equation}
K_p(x'', t; x', 0) = \int \chi_p\left( -\frac{1}{\hbar} S_H[x(t)] \right) \mathcal{D}_p[x(t)],
\end{equation} where \(\chi_p\) is a \(p\)-adic additive character (analogous to \(e^{iS_H/\hbar} \) in real path integrals), and \(\mathcal{D}_p[x(t)]\) denotes integration over \(p\)-adic paths.

Notably, such a theory would not assume smooth differentiable manifolds, but would naturally encode: discrete scale invariance, ultrametric distance hierarchies, and non-smooth, fractal-like phase-space structures. Concrete areas where p-adic structures have already appeared include: \(p\)-adic string theory, where tree-level amplitudes admit natural \(p\)-adic formulations, \(p\)-adic models of quantum gravity, suggesting that spacetime at very small scales may have non-Archimedean structure, and \(p\)-adic diffusion models in disordered systems and complex materials.

Pursuing \(p\)-adic or arithmetic quantization would represent a radically different continuation of classical mechanics: instead of collapsing chaotic or fractal classical action structures into smooth probability amplitudes, one would embrace their discrete, hierarchical nature from the start. Such an arithmetic framework could: capture intrinsic classical multiscale structures naturally, encode branching dynamics and multivalued action surfaces, or model chaotic systems without forcing artificial smoothness or linearity. Thus, p-adic quantization offers another possible future for physical theory: a fundamentally different continuation from classical variational collapse, bypassing the rigid linear Hilbert space structures of conventional quantum mechanics \cite{Schreiber2013}.

Standard quantum mechanics models evolution through unitary flows on a Hilbert space, assuming smooth, global structure in both space and time. However, classical mechanics, especially in regimes involving chaos, fractal structures, and multivalued action, suggests a deeper geometric and algebraic complexity \cite{Volovich1987}. An alternative continuation involves lifting the classical phase space itself into the realm of algebraic geometry, treating points, trajectories, and dynamical flows as structured objects within moduli spaces. In this view, phase-space flows are not simple trajectories but morphisms between algebraic varieties or moduli of coherent sheaves. Instead of representing states as points in a phase space \((q,p)\), one could represent them as objects in a moduli space of stable bundles, divisors, or coherent sheaves. Dynamics become deformations within these moduli spaces, governed by algebro-geometric operations rather than ordinary differential equations. Finally, time evolution could be seen as a sequence of deformations or mutations in these higher structures, with "classical trajectories" emerging only as special cases.

Examples where such ideas already appear include: mirror symmetry and moduli spaces in string theory, where classical geometry and complex structures undergo deep transformations, Fukaya categories and homological mirror symmetry, where classical phase-space structures are categorified into algebraic objects, and Donaldson–Thomas invariants and deformation theory \cite{DonaldsonThomas1998}, where counting stable objects corresponds to deep structural invariants of moduli spaces.

Pursuing algebro-geometric dynamics would represent a continuation of classical mechanics where multivaluedness, branching, chaotic behavior, and topological transitions are not smoothed away, but explicitly incorporated into the fabric of dynamical evolution. Singularities, bifurcations, and phase transitions would naturally correspond to jumps or stratifications within the moduli space \cite{Kontsevich2003}. Algebro-geometric deformations provide another radically different alternative to the linear Hilbert space structure of standard quantum mechanics: one that retains and organizes classical complexity, rather than projecting it into a rigid stationary framework.

This work helps to motivate a natural research program: to systematically classify and explore all consistent mathematical continuations of classical mechanics beyond stationary-action collapse. We term this prospective field \emph{Post-Hamiltonian Representation Theory}. This program would involve identifying all possible mathematical structures (e.g., nonlinear wave equations, branched Riemann structures, sheaf cohomology flows, p-adic dynamics, algebro-geometric deformations) that can consistently extend classical variational principles without enforcing linearity. Also, classifying continuations based on properties such as symplectic structure preservation, multivaluedness handling, phase-topology evolution, and chaotic sensitivity. Finally, developing new physical theories based on these alternative representations, potentially capturing phenomena inaccessible to standard quantum mechanics.

In this vision, quantum mechanics would appear not as a final theory but as the first, historically contingent solution to the problem of classical variational collapse. Other structures might provide alternative or complementary explanations of physical phenomena, new computational paradigms, or even pathways toward reconciling quantum theory with gravity and spacetime discreteness\cite{Lurie2009}. The Post-Hamiltonian Representation Theory aims to open a new frontier in fundamental physics: the systematic exploration of the full landscape of continuations from classical mechanics, unshackled from the rigid assumptions of standard quantum theory.

\section{Discussion}

A central tension in the standard formulation of quantum mechanics lies in the assumption of unitarity, which presupposes a well-defined separable structure (e.g., tensor-product Hilbert spaces), alongside the empirical centrality of entanglement, which violates separability and exhibits global nonlocal correlations. In conventional interpretations, this tension is absorbed post hoc; entanglement is treated as an emergent phenomenon without being built into the foundational axioms. In contrast, the framework developed here resolves this dichotomy naturally: separability and unitarity arise only in the solvable (i.e., decidable) sector of the classical variational theory. Entanglement is not a paradox but a direct manifestation of the underlying non-separability of Hamilton’s action, projected into the linearized quantum approximation. What appears as a conceptual mismatch in quantum mechanics is reinterpreted as an epistemic artifact; a byproduct of representing fundamentally undecidable, globally coupled dynamics within a decidable and separable formalism.

A direct empirical prediction of our framework, where quantum mechanics is viewed as a linearized, decidable subset of an underlying undecidable classical variational structure, is the systematic deviation from standard quantum coherence and thermalization predictions in systems that approach classical chaos thresholds. Lateral double quantum dots (DQDs), precisely controllable semiconductor structures routinely fabricated and studied, offer an ideal platform for this test \cite{Khomitsky2022}.

In lateral DQDs, electron confinement potentials and interdot tunneling rates are finely tunable using electrostatic gates. Adjusting these parameters allows systematic exploration of transitions between integrable dynamics, such as strongly coupled symmetric dots, and chaotic regimes, including asymmetrically coupled dots or driven dots under varying magnetic fields. This tunability enables an explicit empirical test of our theoretical framework. Notably, DQDs have already been used as platforms to investigate signatures of quantum chaology, including studies of level statistics, nonlinear response, and Floquet dynamics in driven configurations \cite{Marcus1997,Anand2024,PhysRevX.7.011034}.

In the proposed experiment, the quantum dot system would be fabricated using standard gated lithographic methods. By systematically adjusting dot potentials and interdot couplings, one can scan from the regular Coulomb blockade regime into classically chaotic analogs. Quantum states should be prepared to align explicitly with classical trajectories near the integrability-chaos boundary. These states can be initialized as coherent superpositions localized in specific dots using precisely calibrated gate voltages and tailored pulse sequences.

Time-domain measurements can then track electron coherence using microwave or radio-frequency reflectometry to detect charge and spin coherence across the double-dot system. The key observables would be relaxation, decoherence, and revival times of quantum states as functions of dot asymmetry, interdot coupling strength, and applied external fields.

Our framework predicts that quantum states prepared along classical trajectories near the chaos threshold will exhibit anomalously prolonged coherence and distinct coherence revivals, effects that standard quantum mechanics, via the eigenstate thermalization hypothesis, would not anticipate. The observation of such anomalies would provide direct empirical support for the hypothesis that quantum mechanics is a computationally restricted image of a richer yet undecidable classical action geometry.

The broader thesis advanced in this work is that quantum mechanics is not a generalization of classical mechanics, but a projection of its representationally complete but computationally inaccessible variational structure. Hamilton’s type-1 principal function, long regarded as overdetermined and mathematically intractable, is shown to encode an action geometry whose formal undecidability parallels the spectral gap problem in quantum many-body physics. This reframing allows foundational puzzles, such as wavefunction collapse, entanglement, spin, and apparent randomness, to be interpreted as structural artifacts. Our analysis of KAM theory shows that classical predictability is also bounded by arithmetic undecidability, challenging the conventional narrative that quantum mechanics replaces classical determinism with stochasticity. Phenomena such as quantum scars, Leggett inequality violations, and the coherence behavior of chaotic systems like Hyperion, all point to persistent classical structure within quantum regimes. These observations reinforce the core argument that quantum theory does not encompass all of classical mechanics, but instead regularizes and linearizes those parts of it that are computationally decidable. 

One common objection to the framework presented here is that it appears incompatible with the widely held view that atoms and photons are intrinsically quantum objects, discrete particles with no classical analog. However, this view relies on an overly rigid dichotomy between classical and quantum descriptions. In classical mechanics, a generic Hamiltonian system is neither integrable nor ergodic, which permits the existence of bounded, nonergodic orbits, structures that can support long lived, localized states without requiring quantization per se. Such structures provide a natural foundation for atomic stability without invoking quantum indeterminacy as a fundamental postulate \cite{MarkusMeyer1974}. The notion of the photon as an irreducible quantum particle has been seriously challenged. It has been emphasizes that photons are not particles in space-time but rather field excitations that become quantized only upon interaction, and that a 'photon-as-a-particle' is even misleading \cite{scully1997}, or even rejecting the particle conception of photons altogether as unnecessary \cite{lamb1995}. 

These perspectives open the door to reinterpreting atomic and photonic phenomena within a broader, dynamically structured classical framework that includes quantized effects as emergent rather than fundamental. In this light, even gauge forces, such as electromagnetism, can be viewed not as fundamental quantum interactions but as geometrical structures introduced to uphold the principle of stationary action across all frames of reference, ensuring that variational symmetry is preserved in relativistic and accelerating coordinate systems.

We propose that the future of fundamental physics may not lie in further quantization, such as quantum gravity, but in developing a more complete understanding of classical action principles and their undecidable limits. The proposed experimental tests in lateral quantum dot systems offer a first empirical step toward this direction, bridging foundational theory and accessible experimental platforms. The program of Post-Hamiltonian Representation Theory introduced here provides a framework for exploring these continuations beyond the quantum collapse.

\section{Conclusion}

Quantum mechanics has long been celebrated for its empirical accuracy and mathematical elegance, yet its foundational tensions, such as wavefunction collapse, the measurement problem, and the classical-quantum boundary, remain unresolved. In this work, we have proposed that these tensions arise not from any incompleteness in classical mechanics, but from a misunderstanding of its representational scope. We use the term ‘proper subset’ in a formal sense: quantum mechanics corresponds to a computationally tractable, linear projection of classical variational mechanics, specifically those parts where the global action function is decidable, smooth, and regular. The classical framework, governed by Hamilton’s principal function, includes regions, such as those involving chaos or undecidable dynamics, that quantum mechanics cannot fully represent.

At the core of this reinterpretation lies the distinction between Hamilton’s principal function and Jacobi’s integrable reduction. The former, governed by a pair of coupled partial differential equations, is shown to be generically undecidable for arbitrary Hamiltonians, mirroring the spectral gap undecidability observed in many-body quantum systems. The latter, while tractable, represents only a restricted slice of the classical phase space, one that corresponds precisely to the domains in which quantum wave mechanics can be faithfully defined.

This computational perspective is reinforced by number-theoretic constraints in KAM theory, where the long-term stability of classical orbits depends on Diophantine conditions that are themselves undecidable. In chaotic systems such as Hyperion, the mismatch between quantum decoherence predictions and observed classical coherence further points to the survival of action-based classical information beyond what is captured in Hilbert space. Quantum scars, too, defy the expectations of ergodicity and thermalization, suggesting that fragments of classical variational structure persist in the quantum regime. Leggett inequality violations, meanwhile, provide experimental evidence that measurement settings influence outcomes in a manner best explained by time-symmetric boundary constraints, as encoded in Hamilton’s action formalism.

Taken together, these insights support the view that quantum mechanics is not a complete theory of nature, but a tractable resolution of classical overdetermination, an epistemic interface designed to operate within the limits of solvability and computability. The true boundary between classical and quantum is not ontological, but logical. Where classical action becomes undecidable, quantum amplitudes emerge as computable substitutes. Collapse is not a physical process, but a structural response to representational failure.

The reliance of quantum cryptographic protocols on the structural features of quantum mechanics further illustrates the consequences of treating the theory as a representationally constrained subset of classical mechanics. Quantum key distribution (QKD) protocols such as BB84 \cite{Bennett1984} are founded on the linearity of quantum evolution, which ensures unitarity and underpins the no-cloning theorem \cite{Wootters1982,Dieks1982}. These features guarantee that quantum information cannot be copied without detection, and that any measurement by an eavesdropper necessarily disturbs the system in a way that can be observed. However, if quantum mechanics arises from a tractable, linearized projection of a deeper classical theory, one that is nonlinear \cite{Terashima2005} and not necessarily unitary, then these foundational guarantees no longer hold. In particular, violations of linearity would permit, in principle, the cloning of arbitrary states and the extraction of information without observable disturbance. This would render current quantum encryption schemes insecure, not due to flaws in their design, but because their security derives from structural constraints that are not fundamental, but emergent. The robustness of quantum cryptography thus serves as a litmus test for the completeness of quantum theory itself.

Our theoretical framework leads naturally to explicit, testable predictions. Specifically, lateral double quantum dots, systems already widely studied and well-controlled experimentally, present an ideal empirical testing ground. By precisely tuning dot potentials and inter-dot couplings to navigate the boundary between classical integrability and chaos, our theory anticipates observable deviations from standard quantum thermalization and coherence predictions, directly linked to the underlying undecidable classical variational structure.

The observation of anomalously prolonged coherence times and robust coherence revivals in such quantum dot configurations would offer clear and compelling evidence supporting the claim that quantum mechanics emerges as a computationally tractable subset of an otherwise classical undecidable dynamics. Such experimental confirmation would not merely support our foundational theoretical claims but could also provide practical routes toward novel quantum technologies that exploit these robust coherence regimes.

By reframing quantum theory as a solvable embedding of a classically undecidable, non-separable action structure, we recover unitarity and separability as emergent features rather than postulates. This unifies the explanatory treatment of unitarity, entanglement, and measurement without invoking metaphysical paradoxes, and invites a re-examination of quantum foundations from the standpoint of classical computational intractability.

Ultimately, testing these predictions through lateral double quantum dots experiments will bridge foundational theory and practical experimentation, offering a unique opportunity to directly probe the deep logical and structural boundaries that separate classical and quantum descriptions of nature.

The future of foundational physics may not lie in quantizing classical systems, but in extending the classical variational framework itself, through multivalued action functions, cohomological methods, arithmetic deformations, or nonlinear wave structures. We propose a broader program in \textit{Post-Hamiltonian Representation Theory}, aimed at systematically exploring these continuations. In this view, quantum theory is not the end of the classical narrative, but a computationally filtered interlude in a deeper geometric story still unfolding. 

\appendix*
\section{Formal Reduction from the Halting Problem to the Solvability of Hamilton's Principal Function}

We provide a more rigorous version of our claim that the global solvability of Hamilton's principal function \( S_H(q, t; q_0, t_0) \) is undecidable. This reduction is inspired by work in computable analysis and dynamical systems \cite{moore1990,pourel1989,richardson1968}. We define our terms precisely and situate the argument within the framework of classical solution theory for partial differential equations.

\subsection*{1. Formal Problem Statement: \(\Pi_{SH}\)-Global-Solvability}

Let \( H(q, p, t) \in C^\infty(\mathbb{R}^{2n} \times \mathbb{R}) \) be a smooth Hamiltonian function.

Define the boundary value PDE system for \( S_H: \mathbb{R}^{2n+2} \to \mathbb{R} \) as:
\[
\frac{\partial S_H}{\partial t_0} = -H(q_0, \nabla_{q_0} S_H, t_0), \quad 
\frac{\partial S_H}{\partial t} = +H(q, \nabla_q S_H, t)
\]

We ask: does there exist a solution \( S_H \in C^2(\mathbb{R}^{2n+2}) \) satisfying this system for arbitrary \( (q, t; q_0, t_0) \)? We denote this decision problem as \( \Pi_{SH} \).

\subsection*{2. Reduction from the Halting Problem}

Let \( M \) be a Turing machine and \( x \in \{0,1\}^* \) its input. Define \( \phi_M \): ``Does \( M \) halt on input \( x \)?'' This is undecidable.

We will construct a smooth Hamiltonian \( H_M(q, p, t) \) such that:
\begin{itemize}
  \item If \( M(x) \) halts, then the associated \( S_H \) exists globally and smoothly.
  \item If \( M(x) \) does not halt, then \( S_H \) is not globally twice differentiable, due to the formation of caustics or discontinuities in \( \nabla_q S_H \).
\end{itemize}

\subsection*{3. Construction of the Hamiltonian \( H_M \)}

Following \cite{moore1990}, we construct \( H_M \in C^\infty \) that simulates a Turing machine \( M \) as a dynamical system. The phase space encodes:
\begin{itemize}
  \item Tape contents as \( q_1, \dots, q_k \)
  \item Head position as \( q_h \)
  \item Internal state via smooth encodings (e.g., bump functions or submanifold labels)
\end{itemize}

The dynamics are encoded such that:
\begin{itemize}
  \item Each computation step corresponds to a Hamiltonian update
  \item Halting corresponds to convergence to a fixed point in phase space
  \item Non-halting results in an unbounded or recurrent flow in configuration space
\end{itemize}

The explicit construction (omitted for brevity) can be found in \cite{moore1990}, which proves the existence of such \( H_M \in C^\infty \).

\subsection*{4. Implications for the Solvability of \( S_H \)}

We interpret \( S_H(q, t; q_0, t_0) \) as the classical action integral for trajectories governed by \( H_M \). Then:
\begin{itemize}
  \item If \( M \) halts, the associated trajectory remains regular and finite. Hence, \( S_H \in C^2 \) exists globally.
  \item If \( M \) does not halt, the induced flow leads to trajectory intersections (caustics), where \( \nabla_q S_H \) is no longer continuous.
\end{itemize}

This is consistent with the failure of classical solution structure in regions where Hamilton-Jacobi characteristics cross, as formalized in microlocal analysis \cite{ArnoldMechanics}.

\subsection*{5. Formal Statement of Reduction}

Define:
\[
\phi_M \iff \text{``M halts on input } x\text{''}
\]
\begin{align*}
\Pi_{SH}(H_M) \iff\ &\text{``There exists } S_H \in C^2\text{ such that} \\
&\text{it solves the PDEs from } H_M\text{''}
\end{align*}

Then:
\[
\phi_M \iff \Pi_{SH}(H_M).
\]

This establishes a many-one reduction from the Halting Problem to \( \Pi_{SH} \). Therefore, since \( \phi_M \) is undecidable, \( \Pi_{SH} \) is undecidable.\qed

\subsection*{6. Conclusion}

Determining the global \( C^2 \) solvability of \( S_H \) for arbitrary smooth Hamiltonians is undecidable. This connects computability theory with the variational structure of classical mechanics and reinforces the paper’s central claim that quantum mechanics emerges as a tractable projection of a richer, computationally inaccessible classical framework.

%\bibliographystyle{unsrt}
%\input{main.bbl}
%\bibliography{main}

\end{document}